\documentclass[twoside, 12pt, fleqn]{amsart}
\usepackage{booktabs,amsmath,amsfonts,amssymb,graphicx,bm}
\usepackage{bbold}
\usepackage{verbatim}
\usepackage[authoryear]{natbib}
\usepackage{relsize}
\usepackage{times}
\usepackage{dsfont}
\usepackage{subfig}
\usepackage{etoolbox}
\usepackage{float} 
\usepackage{epigraph}
\usepackage{ushort}
\usepackage{hyperref}
\usepackage{xcolor}
\usepackage{bbm}
\usepackage{tikz}
\usepackage[utf8]{inputenc}
\usepackage[T1]{fontenc}
\usetikzlibrary{snakes,arrows,calc, patterns, positioning, shapes.geometric, decorations.pathreplacing,decorations.markings}
\usepackage{relsize}
\usepackage{pgfplots}
\usepackage{nccmath}
\usepgfplotslibrary{fillbetween}
\DeclareMathOperator{\sgn}{sign}
\usepackage{titlesec}
\titleformat{\section}
{\color{gray!100}\sffamily\large}
{\color{gray!100}\thesection.}{.5em}{}
\titleformat{\subsection}[runin] 
{\color{gray!100}\normalfont\rmfamily\itshape}
{\color{gray!100}\thesubsection.}{.5em}{}[.]
\titleformat{\subsubsection}[runin] 
{\normalfont\rmfamily\itshape}
{\thesubsubsection.}{.5em}{}[.]

\hypersetup{
  colorlinks   = true, 
  urlcolor     = blue, 
  linkcolor    = blue, 
  citecolor   = blue 
}
\makeatletter
\def\@settitle{\begin{center}%
  \baselineskip14\p@\relax
    \bfseries
    \normalfont\Large\textbf
  \@title
  \end{center}%
}
\usepackage[left=1.25in, right=1.25in, top=1.25in]{geometry}
\usepackage{setspace}
\usepackage[margin=10pt,font=small,labelfont=bf,
labelsep=period]{caption}
\newtheorem{theorem}{Theorem}

\newtheorem{claim}{Claim}
\newtheorem{proposition}{Proposition}
\newtheorem{definition}{Definition}
\newtheorem{obs}{Observation}
\theoremstyle{definition}

\newcommand{\R}{\mathbbm{R}}

\def \qed {\hfill \vrule height6pt width 6pt depth 0pt}

\def\s{\sigma  }

\def\e{\epsilon  }

\usepackage{mathtools}

\begin{document}

\thispagestyle{empty}
\centerline{\color{gray!100}\sffamily\large\textbf{Strategic Observational Learning}}

\vspace*{0.2in}

\centerline{Dimitri Migrow\footnote{Migrow: University of Edinburgh, School of Economics. I thank Debraj Ray for many discussions going back to 2013, when we started work on this problem together. The initial conceptualization of the project in its present form and some of the arguments in their present form are developed jointly with him. For further helpful comments and suggestions I thank Pablo Beker, Matthew Elliott, Mira Frick, Daniel Garrett, Olivier Gossner, Daniel Gottlieb, Peter J. Hammond, Alexander M. Jakobsen, Sam Kapon, Aditya Kuvalekar, Gilat Levy, John Moore, Alessandro Pavan, Marcin Peski, Herakles Polemarchakis, Fedor Sandomirskiy, Aidan Smith, Philipp Strack, Dan Quigley, Colin Stewart, Andy Zapechelnyuk and Weijie Zhong.}}

\vspace*{0.2in}

\centerline{\footnotesize{January 2023}}

\textbf{Abstract}. We study learning by privately informed forward-looking agents in a simple repeated-action setting of social learning. Under a symmetric signal structure, forward-looking agents behave myopically for any degrees of patience. Myopic equilibrium is unique in the class of symmetric threshold strategies, and the simplest symmetric non-monotonic strategies. If the signal structure is asymmetric and the game is infinite, there is no equilibrium in myopic strategies, for any positive degree of patience. \\


\renewcommand*{\thefootnote}{\arabic{footnote}}

\section{Introduction}

Learning by observing the choices of others is an important force that shapes societal outcomes. Such learning is ubiquitous: we learn, and form our opinions by observing which occupations others choose, which fashion trends they follow, which new technologies they adopt and which they dismiss. 
In many of such contexts, choices are made repeatedly creating a feedback loop in the learning process. We make choices that depend on observed actions of other economic agents. At the same time, our choices impact beliefs and future actions of the others which in turn determines the amount of information that we can elicit from others in the future.


One of the key questions in social learning is, how do forward-looking agents learn in settings where actions are taken repeatedly? We still, however, lack an explicit characterization of learning behavior by forward-looking agents in such settings.\footnote{Most of the literature on social learning in repeated-action settings assumes myopia where agents fully discount their future payoffs and optimally choose their ``currently'' preferred actions \textcolor{blue}{\citep{parikh1990communication,gale2003bayesian,harel2021rational,keppo2008optimal,bala1998learning,
razin2021misspecified,frick2020misinterpreting}}. A smaller literature studying forward-looking agents provides important insights into asymptotic properties of information aggregation and agreement between the agents. \textcolor{blue}{\cite{mossel2015strategic}} and \textcolor{blue}{\cite{mossel2020social}} study asymptotic properties of learning and agreement in networks. More recently, \textcolor{blue}{\cite{huang2021learning}} characterize the speed of learning by many forward-looking agents in networks showing failure of information aggregation.} Such settings are difficult to analyze because, different to myopic settings, the behavior of forward-looking agents depends on higher-order beliefs. 

In this paper we are able to explicitly characterize the behavior of forward-looking agents within a tractable model of social learning. As we explain further below, the key feature of our environment is the symmetry between the agents in terms of the information structure that determines their private information. 


Importantly, we show that in our environment of social learning, forward-looking agents behave myopically for \textit{any} degrees of patience. Moreover, the equilibrium in myopic strategies is the unique equilibrium within a family of symmetric strategies that allows for both monotonic and the simplest non-monotonic strategies. This finding may seem surprising as one might expect that forward-looking players have an incentive to experiment strategically and sacrifice their current-period payoffs to learn faster. 

To explain the reasoning behind our results, we first sketch the model. Two players simultaneously and repeatedly choose one of the binary actions in each period. The available actions are $-1$ and $+1$, and the game can be finite or infinite in duration. The unobserved and persistent state is distributed on the real line, with the mean normalized to 0. The state is drawn at the beginning of the game; subsequently, each player receives a single noisy and informative private signal about the state. The support for the signals is the real line, and the signals are iid. The stage payoff for each player is a product of the realized state and her action. That is, each player's current-period payoff is maximized if she correctly matches the sign of the state with the sign of her action: if a player believes that the state is positive, the ``currently'' preferred action is $+1$ and otherwise it is $-1$. The payoffs are only observed at time ``infinity''.

 There are three central features of the model. Actions are coarse relative to signals -- the players cannot fully reveal their signals through actions. Relatedly, signals cannot be directly communicated and the only way to learn other player's private information is by observing their actions. Finally, the choice of actions is payoff-relevant.

As indicated above, we distinguish between symmetric and asymmetric settings. A symmetric setting satisfies two conditions: (1) the prior is symmetric around the mean and (2) players' signals are drawn from the same signal family where ``positive states'' have the same stochastic properties as ``	negative states''. \emph{The players themselves need not be symmetric}: they can differ ex ante as their discount factors may differ, and also differ ex post after they observe the realized signals.

Our main result shows that in the symmetric setting myopic equilibrium always exists, for any profile of discount factors. 
Specifically, we provide two very different models of off-path beliefs, and show that the myopic equilibrium exists under both belief systems. 

In the symmetric setting, myopic equilibrium is unique both in the class of symmetric threshold strategies and the simplest symmetric non-monotonic strategies. The latter strategies allow the players to use two distinct thresholds within the same period.

To understand the above results consider, first, symmetric threshold strategies. If a player uses a non-myopic threshold in some period, the other player must use a symmetric non-myopic threshold in the same period. In this case, however, the threshold types of both players perfectly observe all payoff-relevant information at the end of the period independent of their own action in that period. As a result, both the threshold types of the players, and the types arbitrarily close to the threshold types who are expected to play non-myopically, have a strict incentive to deviate and to choose their respective myopic actions.

Second, consider the simplest symmetric non-monotonic strategy where after some histories the players use two distinct thresholds in the same period: they take the same action below the first, and above the second threshold, while choosing a different action between the two thresholds. We show that under such strategies the players must have a dominant action for some of their types below the lowest threshold, and a \textit{different} dominant action for some of their types above the second threshold. Therefore, such a strategy can never be part of an equilibrium.

We then show that in the class of arbitrary (i.e., not only symmetric) threshold strategies, asymptotically -- as the time horizon goes to infinity -- only the myopic equilibrium can fully aggregate private information. This information aggregation property of the myopic play is based on the following results proven in the paper. 
First, under myopic play disagreement cannot continue forever. Second, agreement is always based on weakly dominant strategies. The two results imply that eventually the players agree and under an agreement they choose the same actions as what they \textit{would have chosen} under complete information.
Finally, any other threshold strategy fails to fully aggregate private information: there always exist measurable sets of player types where players choose a wrong action (i.e. a different action compared to the one they would have chosen under complete information).

Finally, we study asymmetric environments and show that the myopic equilibrium does not exist in infinite games. Here, an agreement under myopic play does not exhaust all learning potential as it is the case in the symmetric setting. In fact, in the asymmetric setting already in the second period of the game there are types of each player who have a profitable deviation from the myopic play. Their current-period losses from a deviation could be held arbitrarily small, while at the same time they expect to have strictly positive gains from learning in future periods. 


\medskip

\noindent\textbf{Related Literature}. We are not aware of any other papers that explicitly characterize equilibrium behavior within a repeated-action setting of social learning with forward-looking players.

There is an important literature on asymptotic properties of learning and agreement. The perspective here is that of a (long-run) steady-state as in \textcolor{blue}{\cite{aumann1976agreeing}} who shows that in a common-prior setting, if players' posteriors are common knowledge, then the players must agree on their assessment of an event. \textcolor{blue}{\cite{geanakoplos1982we}} take a dynamic perspective and analyze the sequential exchange of posteriors that eventually leads to an agreement. \textcolor{blue}{\cite{mossel2020social}} adopt the perspective of a static equilibrium, a long-run steady-state similar to \textcolor{blue}{\cite{aumann1976agreeing}}, and investigate whether forward-looking players learn the dispersed information, and whether they agree on the common course of action. They show that in a variety of canonical models within social learning literature, learning and agreement happens asymptotically.

\textcolor{blue}{\cite{mossel2020social}} build upon their earlier work, \textcolor{blue}{\cite{mossel2015strategic}}. The latter paper studies agents' asymptotic behavior in observation networks. In particular, it studies how the geometry of networks affects asymptotic agreement and learning. The main take-away here is that if there are no players who are very influential (i.e., if the networks are sufficiently egalitarian), learning and agreement happen in the long run. I.e., Bayesian players would eventually agree on the right action, given the realized state of the world. Interestingly, this finding holds qualitatively also for boundedly rational players \textcolor{blue}{\citep{bala1998learning}}.

We complement the above literature on asymptotic properties of learning and agreement in that we study how the players arrive at their beliefs. In this sense, while the above literature takes the perspective of \textcolor{blue}{\cite{aumann1976agreeing}}, we take the perspective similar to \textcolor{blue}{\cite{geanakoplos1982we}} and focus on the explicit learning path of forward-looking players.\footnote{For a broader discussion of the literature on social learning in a repeated-action setting we refer the reader to an excellent survey by \textcolor{blue}{\cite{bikhchandani2021information}}.}

We note here that \textcolor{blue}{\cite{mossel2015strategic}} build upon \textcolor{blue}{\cite{rosenberg2009informational}} who establish an imitation principle that states that asymtotitcally each agent does as good as their neighbor, whom he observes and whom he can imitate. A similar principle is discussed in \textcolor{blue}{\cite{golub2017learning}}. \textcolor{blue}{\cite{rosenberg2009informational}} build upon \textcolor{blue}{\cite{gale2003bayesian}} who analyze learning in networks with myopic players.

Our model assumes coarse action space relative to the signal space. As a result, the players cannot reveal their signals through actions. In a setting with sufficiently rich signal space, \textcolor{blue}{\cite{mueller2014does}} shows that information is aggregated independently of the underlying network configuration.

There is a complementary literature that departs from the perfect rationality assumption. For example, \textcolor{blue}{\cite{frick2020misinterpreting}} show that in a repeated-action setting of social learning even a small misperception of the correct model by the agents can lead to stark negative effects on learning. \textcolor{blue}{\cite{bohren2021learning}} study multiple forms of misspecification. In their paper, ``decoupled learning'' as in \textcolor{blue}{\cite{frick2020misinterpreting}} is absent so that aggregation of dispersed information is robust to small misspecifications. Further studies of learning by boundedly rational agents in a repeated-action framework include  \textcolor{blue}{\cite{degroot1974reaching}}, \textcolor{blue}{\cite{demarzo2003persuasion}} and \textcolor{blue}{\cite{golub2010naive}}.

A complementary literature on herding studies social learning in settings where each agent moves only once \textcolor{blue}{\citep{banerjee1992simple,bikhchandani1992theory,acemoglu2011bayesian,lobel2012social,
smith2000pathological}}. This literature characterizes information cascades where optimal actions could be chosen independent of agents' own private information.

\section{Model}
Two players, $a$ and $b$, simultaneously and repeatedly choose binary actions $z$ from the set $\{-1,+1\}$ in each period $t=0, 1, 2 \ldots \bar{T}$, where $\bar{T}$ may be finite or infinite. There is an unobserved, payoff-relevant state $x\in\mathbbm{R}$, with continuous and full-support prior density $f(x)$, and with mean set to 0.

\noindent \textit{Payoffs}: If agent $i$ takes action $z_i$, her stage payoff is $z_ix$. So the signs of action and state need to match for a positive payoff, and the larger the absolute value of the state, the larger the absolute value of the payoff. For any sequence of actions $\{ z_i(\tau )\}$ starting from $t$, $i$'s discounted continuation payoff is $\mathbbm{E}_t [\sum_{\tau =t}^{\bar{T}} \delta_i^{\tau -t} z_i(\tau)x]$, where the expectation is taken relative to information available at date $t$. 
As is customary, we assume that the players experience their payoffs ``at infinity."


\noindent \textit{Signals and information}:  At the start of their interaction, each player $i$ receives a \emph{one-time} private signal (or type) $s_i$  about $x$, (conditionally) independently drawn with {full-support  density} $g_i(s_i|x)$, with mean $x$.   We will often use the term ``type" to describe these initial signals.

When we use $i$ to index either player (say $b$), we will refer to the other player as $j$ (say $a$). Let $E_i(s_i, s_j)$ be the posterior expectation of $x$ (using Bayes' Rule) under  $s_i$ and $s_j$. Of course, $E_a(s_a, s_b) = E_b (s_b, s_a)$  but there is no additional presumption of symmetry.
We can extend $E_i$ to sets as follows. Suppose that $j$'s signal is {believed} to lie in the set $S_j$. Then player $i$ with this knowledge computes the posterior mean of $x$ by integrating $E_i(s_i, s_j)$ as $s_j$ ranges with its conditional density projected onto the restricted domain  $S_j$. With some abuse of notation we also refer to this function as $E$. We assume that:

[E.1] {For any set $S_j$}, $E_i(s_i, S_j )$ is increasing in $s_i$, with $\lim_{s_i \to \infty} E_i(s_i, S_j ) > 0$ and $\lim_{s_i \to -\infty} E_i(s_i,S_j ) < 0$.

%
%
%
%

As an example: $x$ is drawn from $N(x_0,\sigma^2)$ with $x_0=0$. The agents observe signals $s_a=x+\epsilon_a$, $s_b=x+\epsilon_b$, where $\epsilon_a\sim N(0,\sigma_a^2)$ and $\epsilon_b\sim N(0,\sigma_b^2)$, so that:
\[x\sim N(x_0,\sigma^2),\ s_a|x\sim N(x,\sigma_a^2),\ s_b|x\sim N(x,\sigma_b^2).\]
For $S = [a, b]$ and $m \in S$, call $[m, b]$ a \emph{left truncation} and $[a, m]$  a \emph{right truncation} of $S$. (A right truncation is strict if $m \neq b$, and likewise for left truncations.) It is easy to see that $E(s, S') > E(s, S)$ if $S'$ is a strict left truncation of $S$, with the opposite inequality if $S'$ is a strict right truncation.

\emph{Strategies and equilibrium}: The past play of all actions up to any date $t$  --- the (public) $t$-history --- is commonly observed, with any singleton standing for the null history at $t=0$. A \emph{strategy} for player $i$ specifies $z \in \{ -1, +1 \}$ at each $t \geqslant 0$ conditional on her type and the $t$-history: all actions taken prior to $t$. Our solution concept is perfect Bayesian equilibrium. We  discuss restrictions on ``off-path" beliefs below.

\noindent \emph{Central features}: While the model is highly stylized, three features are critical, while the rest can be relaxed or generalized substantially. The first is that individuals cannot communicate their signals directly. Second (and relatedly), the set of available actions is coarse relative to the signals, so that signals cannot be fully inferred from the actions  --- the assumed continuum of states and signals is meant to approximate this. The third feature is that  actions, while directly payoff-relevant, also have signaling value.

\section{More on Strategies and Beliefs}
\subsection{Histories} A history at date $t$ --- or a \emph{$t$-history} $h_t$ --- is a full specification of actions taken by both players up to the start of (but not including actions at) date $t$. By convention, there is just some arbitrary singleton history at 0.

\subsection{Threshold Strategies} A \emph{threshold strategy} for  $i$ is a (history-dependent) sequence $\{ \mu_i(t)\}$, with the interpretation that $i$ chooses $+1$ at date $t$ if $s_i > \mu_i(t)$, and $-1$ if $s_i < \mu_i(t)$. The specification at $s_i = \mu_i(t)$ is arbitrary, and we ignore it throughout.

\subsection{Myopic Strategies} 
A  \emph{myopic strategy} chooses actions $z_i(t)$ at every date to maximize current payoff $z_i(t)\mathbbm{E_t}(x)$. Under (E.1), a myopic strategy is a threshold strategy.

\subsection{Beliefs} As play unfolds under some presumed strategy profile, each player $i$ will have a belief (about $j$'s type) supported on a closed subset of types for $j$. At the start of each date $t$, {before} actions have been chosen, let $S_j(t)\subseteq \mathbbm{R}$ denote this \emph{belief set}. To illustrate, at the start of date of $0$, $a$ believes that any signal is possible for $b$, so $S_{b}(0)= \mathbbm{R}$. Suppose that $b$ is believed by $a$ to use myopic strategies, and $b$ chooses $+1$. Then $S_{b}(1)=[0,\infty)$. At date $1$, suppose that $b$ is believed to use a threshold  $\mu_{b}(1) >0$. If $b$ chooses $-1$ at date $1$, then $S_{b}(2)= [0, \mu_b(1)]$, and so on. 

Suppose that $i$ believes that $j$ uses a particular strategy. Play is \emph{presumed (by $i$) to be on-path} if no action by $j$ contradicts $i$'s belief.

\begin{obs}
With $j$ believed to use threshold strategies, $i$'s belief sets after every presumed on-path history are  intervals, independent of $i$'s type.
\label{obs:thresh}
\end{obs}
\begin{proof} $S_j(0) = \R$, so the claim is true at date 0. Inductively, suppose that for some $t \geqslant 0$, the claim is true. Consider date $t+1$ and any presumed on-path $(t+1)$-history at the start of $t+1$. Consider its $t$-subhistory. Then $S_j(t)$ is an interval independent of $i$'s type. Let $\mu_ j(t)$ be $j$'s threshold at $t$. Because the $(t+1)$-history is presumed on-path, $j$ can only play $-1$ if $\mu_ j(t) \ge \inf S_j(t)$. If she does,  $i$'s belief set becomes
\[
S_j(t+1) = (-\infty, \mu_ j(t)] \cap S_j(t) = [\inf S_j(t), \min\{ \mu_ j(t),\sup S_j(t)\}],
\]
an interval independent of $i$'s type. Parallel arguments apply if $j$ plays $+1$ on-path.
\end{proof}
Note that while the on-path belief sets held by $i$ are independent of $i$'s type, her belief \emph{distribution} over the type of $j$ will certainly depend on $i$'s realized type --- after all, the types are correlated  via the true state. But with pure strategies that  belief distribution can be calculated given only the going belief set. With that in mind,
for any belief set $S_j$, define $\s_i (S_j)$ by the unique solution in $m$ to $E(m_i, S_j) = 0$. For any singleton belief set $S_j = \{ z \}$, we write $\sigma _i(\{ z\})$ simply as $\sigma_i (z)$.
\begin{obs} 
For any belief set $S_j$ at any date $t$:

(i) Player $i$'s myopic strategy dictates that she plays $-1$ if $s_i < \s_i (S_j)$,  and $+1$ if $s_i > \s_i (S_j)$.

(ii) If $s_i < \s_i (\sup S_j)$, then it is a dominant action for $i$ to play $-1$.  Likewise, if $s_i > \s_i (\inf S_j)$, then it is a dominant action for $i$ to play $+1$. That dominance continues for all subsequent on-path histories.
\label{obs:dom}
\end{obs}   
\begin{proof}
Part (i) is a direct consequence of Assumption E.1 and the definition of $\sigma_i$. For part (ii), 
suppose that $s_i < \s_i (\sup S_j)$. Note that player $i$ holds certain  beliefs that player $j$'s signal is no larger than $\sup S_j$. But then, player $i$'s expectation of the state is negative, and she is certain that it will remain that way in the future, because presumed on-path belief sets can never expand. Therefore her play of $-1$ must be dominant. A parallel argument holds when $s_i > \s_i (\inf S_j)$.
\end{proof}
Of course, a player could be confronted by an \emph{unexpected action}, one incompatible with her current beliefs about types and strategies. Then belief sets could alter in all sorts of ways, and in particular, previously dominant actions could be rendered un-dominant. We will return to this issue below.

\section{Some Properties of Myopic Play}
Under myopic play, once two players agree at date $t$ --- that is, they take the same action --- they must agree forever after, and if they disagree, there is always room for further updating. In short, barring measure-zero realizations of types, they cannot disagree forever. We do not provide these results for their own sake as they are closely related to well known arguments about agreeing to disagree, but rather as self-contained lemmas for the propositions that we shall later establish.

It is convenient to work in a slightly more general setting,  with arbitrarily assigned initial belief sets.  (In actuality,  initial belief sets are pegged at $\R$ for either player.)  At any initial date, a configuration under myopic play is given by $\{ S_a, S_b, m_a, m_b \}$ such that each  $S_j$ is an interval representing the belief of $i$ about $j$'s type, and each $m_i$ is to be interpreted as player $i$'s \emph{myopic threshold} at that date. By Observation \ref{obs:dom}(i), $m_i = \s_i (S_j)$.
Given this configuration, \emph{actual} myopic play at that date depends on the types $(s_a, s_b) \in S_a \times S_b$. There is \emph{agreement} if both players choose the same action and \emph{disagreement} if they choose different actions. The following observation charts the course of myopic play. In what follows and in the rest of the paper, ``will happen," ``will agree," cannot happen," etc., all refer to probability one events. 
\begin{obs} Assume myopic play relative to some initial  configuration at some starting date $t$, given by two non-degenerate closed intervals $S_a$ and $S_b$, and that 
\begin{equation}
m_i \in \textnormal{\mbox{Int }} S_{i} \mbox{ for at least one player } i.
\label{eq:intm}
\end{equation}
Then: 

(i) if $a$ and $b$ agree at $t$, they will agree for all $s \ge t$, with the same actions as at date $t$.

(ii) if they disagree, then either they will agree tomorrow, or the updated configuration has  the properties listed at the start of this Observation,  including (\ref{eq:intm}).

(iii) Disagreement cannot continue forever. 
\label{obs:myopic}
\end{obs}
\begin{proof}
(i) Suppose that the two players agree by each playing $+1$ (the argument for $-1$ is obviously symmetric). Then $m_i < \sup S_i$ for each $i$, and the updated belief sets are given by
\begin{equation}
S'_i = [\max\{ \inf S_i, m_i \}, \sup S_i] \mbox{ for } i = a, b.
\label{eq:agree}
\end{equation}
Then $S'_{i}$ is a  left truncation of $S_i$ for each $i$,  so
\begin{equation}
m'_i \leqslant m_i \mbox{ for } i = a, b.
\label{eq:m1}
\end{equation}
Combining (\ref{eq:agree}) and (\ref{eq:m1}), we see that no more belief updates will occur and that both players will play $+1$ thereafter at all future dates.

(ii) Suppose that the two players disagree by playing $\{-1, +1\}$ (the argument for $\{ +1, -1\}$ is again symmetric). This can only happen with positive probability if $m_a > \inf S_a$ and $m_b < \sup S_b$. Suppose without loss of generality that  (\ref{eq:intm}) holds for player $b$. Then the updated belief sets are given by
\begin{equation}
S'_a = [\inf S_a, \min\{ m_a, \sup S_a \}] \mbox{ and } S'_b = [m_b, \sup S_b].
\label{eq:disagree}
\end{equation}
 Then $S'_b$ is a strict left truncation of $S_b$, and so $m'_a < m_a$.  Now observe that  $E(S_a, m_b)=0$, which implies from (\ref{eq:disagree}) that $E(S_a, S'_b) > 0$. But $E(m'_a , S'_b )= 0$, so we must conclude that $m'_a < \sup S_a$, and therefore we have
\begin{equation}
m'_a <  \min\{ m_a, \sup S_a \}.
\label{eq:amin}
\end{equation}
 With (\ref{eq:amin}) in mind, we see that just two cases are possible: 
 
Case ii-a. $m'_a > \inf S_a = \inf S'_a$. Then using (\ref{eq:disagree}) and (\ref{eq:amin}), we see that $m'_a \in\mbox{Int } S'_a$ and we are done.
 
Case ii-b. $m'_a \leqslant \inf S_a = \inf S'_a$. Now recall that $S'_a$ is a right truncation of $S_a$, so $m'_b \geqslant  m_b$. If equality holds, then, using (\ref{eq:disagree}), we have $m'_a \leqslant \inf S'_a$ as well as $m'_b = m_b =  \inf S'_b$, so tomorrow both players must agree (and play +1) with probability one.

On the other hand, if equality does not hold, then $m'_b > m_b =  \inf S'_b$. In this subcase, using  $m'_a \leqslant \inf S'_a$, we have $E(m'_a, m'_b) \leqslant E(S'_a, m'_b)=0$. But we also know that $E(m'_a, S'_b) = 0$. Combining and using the non-degeneracy of $S'_b$, we must conclude that $m'_b < \sup S'_b$, which implies that $m'_b \in\mbox{Int } S'_b$ and once again we are done.

(iii) Suppose there is perpetual disagreement. Then starting from the configuration $(S_a, S_b, m_a, m_b) = (\R , 0, \R , 0)$, which satisfies (\ref{eq:intm}), let $\{ S_a(t), S_b (t), m_a(t), m_b(t)\}$ be the succeeding sequence of belief sets. At each date (\ref{eq:intm}) is satisfied. So there is some player $i \in \{ a, b\}$ and a subsequence of dates $t_k$ such that $m_i(t_k) \in \mbox{Int } S_i(t_k)$ for all $k$, and by the updating equation (\ref{eq:disagree}) for disagreement, $\cap_{k} S_i(t_k)$ is a \emph{singleton} --- the nested sequence $S_i(t_k)$ always shrinks. By Observation \ref{obs:thresh}, these sets depend on the history of play but --- controlling for that history --- not on the players' types. And yet for disagreement to be perpetual with positive probability, it must be that this singleton limit equals the true type of $i$, which can only happen with probability zero.
\end{proof}

\section{Symmetry}
\label{section:sym}
Our model is \emph{symmetric} if the following conditions hold:

(i) The prior on the state  is symmetric about its  mean of 0: $f(z)=f(-z)$ for all $z$.

(ii) Player types are drawn from the same signal family $\{ g(s|x)\}$, and for each $(x,s)$, 
$g(s|x)=g(-s|-x)$.

Thus symmetry holds if ``positive states" have exactly the same stochastic properties as ``negative states." Notice that this does not impose full symmetry on the players themselves, certainly not ex post after they receive their draws, and not even ex ante (they could have different discount factors).

{Neither perfect Bayesian equilibrium nor its known refinements restrict off-path beliefs in our context. We note two  opposing but reasonable restrictions on such beliefs.

\textbf{B.1.}  \emph{Belief Inertia}. Upon observing an unexpected action by $j$ that is incompatible with on-path play, player $i$  disregards that action and does not update her beliefs at all.

\textbf{B.2.} \emph{Belief Reset}. Upon observing an unexpected action by $j$ that is incompatible with on-path play, player $i$ discards every recent belief update, working backwards, until she finds the most recent prior at which the current move can be viewed as on-path. She then updates that prior with the current move. 

The two variants live at extreme ends, in the following sense: (1) assumes that  the latest unexpected move is a  tremble, while (2) fully respects  the latest move, with the immediately preceding moves viewed as  trembles.

\begin{theorem}
Under symmetry, myopic play is an equilibrium, and off-path beliefs can be taken to satisfy (B .1) or (B.2). 
\label{myopic_equil}
\end{theorem}}
\begin{proof}
We refine Observations \ref{obs:thresh} and Observation \ref{obs:myopic} for the symmetric case.

\textbf{Step 1.} Under the common belief that play is myopic, consider any on-path $(t+1)$-history with no common action ever taken for $\tau \leqslant t$. Then $m_a(t+1) = - m_b(t+1)$, $S_a(t+1)$ and $S_b(t+1)$ are non-degenerate intervals with $S_a(t+1) = -S_b(t+1)$, and $m_i(t+1) \in \mbox{Int }S_i(t+1)$ for $i = a, b$.

The symmetry of the configuration at every date (conditional on past disagreement) follows by a simple recursive argument that uses the symmetry of the model. Additionally, in the symmetric case, notice that if $S_j$ is any nondegenerate belief interval and $m_i = \s_i (S_i)$, then $m_i \in \mbox{Int } (-S_j)$. Because $S_a(t+1) = -S_b(t+1)$, it follows that $m_i(t+1) \in \mbox{Int }S_i(t+1)$ for $i = a, b$.

\textbf{Step 2.} Under the common belief that play is myopic, consider any on-path $t$-history with no common action ever taken for $\tau \leqslant t$.  Suppose that $j$ plays $z$ at date $t$, and is believed to do so by $i$ as a myopic best response. Then, if $\sgn{[s_i -  m_i(t)]} = \sgn{(z)}$, it is uniquely optimal for $i$ to play $z$ from $t+1$ onward.


Because  there is no common action before $t$, Step 1 implies that $S_j(t)$ is a non-degenerate interval and $m_j(t) \in \mbox{Int }S_j(t)$. Say $j$ plays $z=+1$ at $t$; then $s_j > m_j(t)$ a.s. Additionally, $m_i(t) = - m_j(t)$ and $s_i > m_i(t)$ (given the assumption of the Step and $z = +1$). Therefore
\[
s_i + s_j > m_i(t)  + m_j(t) = 0
\]
almost surely with respect to $i$'s beliefs. This belief cannot change on-path, so it is a best response for $a$ (myopic or not) to take action $+1$ at all periods thereafter. \qed

\textbf{Step 3.} For any $t\geqslant 0$ consider any on-path $t+1$-history in which a common action is played for the first time at $t$. Then myopic strategies prescribe the play of that common action from $t+1$ onwards for both players, and this is an equilibrium regardless of players' discount factors. 

Say the common action is $+1$. Because each $i$ plays $+1$ at date $t$ under her myopic strategy, it must be that  $s_i > m_i(t)$ a.s. Now  the conditions of Step 2 are satisfied, so the conclusion of that Step must hold. So the myopic strategy is a best response for each player regardless of their discount factors. \qed

Returning to the proof of the theorem, suppose that myopic play is not an equilibrium. Then there is a first date $t$ and an on-path $t$-history with the property that  $a$ (say) has a profitable deviation. By Step 3, it must be that the $t$-history in question has no record of common play by the two agents. By Step 1, it must be that at date $t$, agents have thresholds $m_a(t)$ and $m_b(t)$ with $m_a(t) + m_b(t)=0$. Choose player indices so that $m_a(t) \leqslant 0 \leqslant m_b(t)$, and  $s_a < 0 <  s_b$.

\begin{figure}[t!]
\begin{center}
\hspace*{-0.25in} 
\subfloat[{\small Case 1}]{
\includegraphics[width=0.43\textwidth]{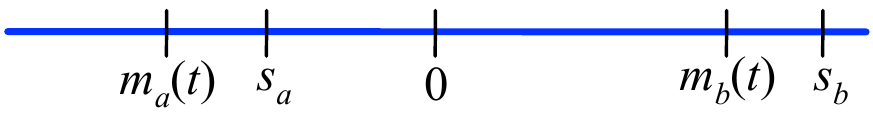}} 
\subfloat[{\small Case 2}]{
\includegraphics[width=0.43\textwidth]{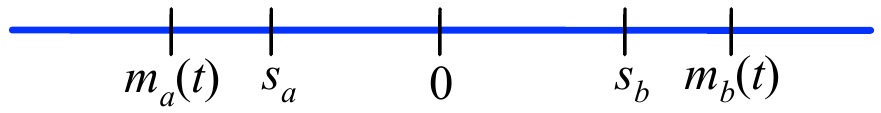}}
\end{center}
\caption{Cases 1 and 2 in the proof of Theorem \ref{myopic_equil}.}
\label{fig:1}
\end{figure}

Suppose that $s_a > m_a(t)$. (An entirely parallel argument will apply if the opposite inequality holds.) Then it is optimal under myopic play for $a$ to choose $+1$. Suppose $a$ plays $-1$ instead. On today's expected payoff she makes a weak loss. To evaluate expected \emph{future} payoffs, consider two cases (use Figure \ref{fig:1} to accompany the discussion):
 
Case 1. Player $b$ (playing the myopic strategy) plays $+1$ at date $t$, and so $s_b \geqslant m_b(t)$.  So Step 2 applies in this case with $i=a$ and $z = +1$, and $a$ therefore optimally plays $+1$ forever afterwards. But $a$ would have done the same in this case \emph{even if she had not deviated} at $t$. So given the expected weak loss at date $t$ from choosing $-1$, her deviation cannot be profitable overall.

Case 2.  Player $b$ (playing the myopic strategy) plays $-1$ at date $t$, and so $s_b \leqslant m_b(t)$. Recall that Player $a$ has played $-1$ also, which is compatible with on-path play, and believed to be so (albeit erroneously) by player $b$. Step 2 now applies to player $b$, with $i=b$ and $z = -1$, who will now proceed to play $-1$ for ever after, assuming that $a$'s future play remains compatible with myopic play.

What has player $a$ learnt from this deviation? She now knows that $s_b \leqslant m_b(t)$ (from the play of $-1$ by $b$) but she could have learnt that anyway, even by not deviating. Consequently, there is no gain for player $a$ deviating to $-1$, \emph{unless she is willing to move again to} $+1$. If she does so at some date $\tau > t$, $b$ is faced with an unexpected move, and adopts a new belief according to Variants B.1 or B.2.

Under Variant B.1, the unexpected deviation is ignored, beliefs are not updated, so player $b$ continues to play $-1$ at all dates. Additional deviations by $a$ make no difference to this outcome. Under Variant B.2, player $b$ will discard the immediately preceding play(s) of $-1$ by player $a$, and update her prior at the start of date $t$ by regarding the latest play of $+1$ as on-path. At date $\tau + 1$, then, player $a$ is back to where she would have been at the start of date $t+1$, had she not made her original deviation. All learning has therefore been pushed back by at least one period, which cannot be profitable, either because of discounting, or the possibility that $T$ is finite (or both).\footnote{Also note that any additional deviation returns player $b$ to playing $-1$ forever, and learning stops again.}
\end{proof}

\subsection*{An Example That Illustrates Theorem \ref{myopic_equil}}
Assume $t=0,1,2,3$. Consider a myopic strategy profile. Under this strategy profile, and neglecting  zero signals, $\sgn (z_i(0))= \sgn (s_i)$ for each $i$. So their actions will reveal the signs of their signals immediately. If $z_a(0)=z_b(0)$, it will be optimal for both players to take this common action in the three remaining periods (as in Step 3 of the proof).

Now suppose (without loss) that  $z_a(0) = -1$ and   $z_b(0) = +1$, so that $s_a<0$ and $s_b>0$. Now consider $a$'s strategy at date $t=1$. If her signal $s_a$ is close to 0, then (knowing that $s_b$ is strictly positive but not its exact location), she will conclude that the expectation of the state is strictly positive, and she will choose $+1$ at date 1. That expectation strictly declines as $s_a$ becomes more negative. Indeed, there must exceed a threshold signal --- call it $m_a(1) < 0$ --- at which her expectation of the state drops to zero. This stems not just from her own reduced signal but also from the resulting downgrade in her prior over $b$'s signal. By exactly the same logic, $b$ has a threshold $m_b (1)> 0$. By the symmetry of the problem, $m_b(1)=-m_a(1)$. 

This symmetry shows that if players take a common action at date 1, then once again they will persist with that action in dates 2 and 3. For instance, if $(+1,+1)$ is played at date 1, this shows that $s_a > m_a(1)$ and $s_b > m_b(1)$, so that the sum of the signals is commonly believed to exceed $m_a(1) + m_b(2)$, which by symmetry is 0. (This is Step 3 of the proof again.)

Otherwise players $a$ and $b$ choose different actions at both dates 0 and 1.  Then players will choose new thresholds in period 2, and there is still no clear-cut action independent of the type (as in Step 1 of the proof). However, in this simple four-period setting it becomes clear that in periods 2 and 3 there is no further incentive to deviate from myopic play.

It remains  to rule out non-myopic play in periods 0 and 1. First consider date 1. Suppose that myopic play calls for $(-1, +1)$, which means that $s_a < m_a(1)$ and $s_b > m_b(1)$. If player deviates to $+1$, player $b$ will see the common action $(+1, +1)$ played at date 1. Nothing about this is off-path from player $b$'s perspective, but she will now wrongly believe that $s_a > m_a(1)$. She will therefore intend to play $+1$ in both periods 2 and 3. Thus player $a$ has learnt nothing about player $b$'s signal over and above what she would have learnt had she simply not deviated --- in addition, she would have done a better job of matching the state (given her own belief that the state has negative expectation). 

Because the game ends in period 4, there is not even any need to impose discipline on off-path beliefs. Player $b$'s mind is already made up for date 2, and even if Player $a$ can surprise her at the end of that date, there is no point in doing do as they will move simultaneously at the final date 3.

However, that discipline does play a role for period-0 deviations. Recall that $s_a < 0$ and $s_b > 0$, so myopic play dictates $(-1, +1)$ in that period. If player $a$ deviates to $+1$, the same logic that we employed for a period 2 deviation ensures that player $b$ will now intend to play $+1$ for all three remaining periods (Step 2 of the proof again). Once more, player $a$ is about to learn nothing, \emph{and} has mismatched the state at date 0.

Player $a$'s only hope of learning anything more about the state is to switch back to playing $-1$ at date 1. Now \emph{this} surprises player $b$ as it is off-path. As posited by us, her reaction  can be of two kinds:

(B.1) Her beliefs are inertial: she ignores the reswitch as a tremble. But then player $b$ will continue to play $+1$ as before in the remaining three periods, and player $a$'s original deviation is unprofitable.

(B.2) She resets her beliefs about Player $a$ by respecting the current move ($-1$) and discarding her pervious update. Now it is as if they are informationally back at the end of period 0 again. Player $a$ could have been in this position by not having deviated to begin with. Additionally, she would have matched the expected state in date 0. Finally, she would have had three periods more to learn about player $b$'s signal, rather than just two after the double deviation.

These informal arguments illustrate the more general Theorem \ref{myopic_equil}.


\begin{definition}
The strategies are symmetric if the following holds. For any history $h(t)=(+1,-1,..)$ let a player of type $s$ choose $z_t\in\{-1,+1\}$. Then, for the ``flipped history'' $h'(t)=(-1,+1,..)$ a player of the ``flipped'' type $s'=-s$ should choose $-z_t$. The same holds after the empty history $h=h'=\emptyset$.
\label{def_sym}
\end{definition}

In the following we define the simplest non-monotonic strategy which we call a \textit{two-threshold strategy} where a player uses two distinct thresholds in the same period after at least one history.

\begin{definition}
Under a two-threshold strategy, after each history $h(t)$ (that includes a null history) player $i$ either plays according to a single (history-dependent) threshold $\mu_i(t)$ defined earlier, or according to two (history-dependent) thresholds $\{\mu_i^1(t),\mu_i^2(t)\}$, $\mu_i^1(t)<\mu_i^2(t)$, with the interpretation that an action $z_i'\in\{-1,+1\}$ is chosen for $s_i<\mu_i^1(t)$ and $s_i>\mu_i^2(t)$, and the other action $-z_i'$ is chosen for $s_i\in(\mu_i^1(t),\mu_i^2(t))$. The specification at $s_i=\mu_i^1(t)$ and $s_i=\mu_i^2(t)$ is arbitrary. There must exist at least one history after which the history-dependent strategy prescribes using two thresholds in the same period.
\label{two-thresh}
\end{definition}


\begin{proposition}
Under symmetry, in a game with discounting $($i.e. $\delta_a\in(0,1),\delta_b\in(0,1)$$)$ myopic equilibrium is the unique equilibrium in the class of symmetric threshold strategies and symmetric two-threshold strategies.
\label{sym_thresh}
\end{proposition}

\begin{proof}
(1) Consider the play in symmetric threshold strategies.

First, in $t=0$ by definition there is a unique threshold at $0$: $m_a(0)=m_b(0)=0$.

Second, consider a history $h_T$ such that in all $t=0,..,T-1$ the play has been myopic and no agreement has yet been reached (i.e. $z_a(t)\neq z_b(t)$ for all $t<T$). At the beginning of $T$ the players correctly believe that $\text{sign}(s_a)\neq \text{sign}(s_b)$. Consider $s_a<0<s_b$ (the proof for $s_a>0>s_b$ is entirely symmetric) and assume that player $b$ uses a non-myopic strategy in the period $T$, say, $\sigma_b(S_a(T))<m_b(T)$ (a symmetric argument applies for the parallel case $\sigma_b(S_a(T))>m_b(T)$). By Definition \ref{def_sym}, player $a$ must use a non-myopic threshold $m_a(T)=-m_b(T)$ which implies that $m_a(T)<\sigma_a(A_b(T))$. However, the type $s_a'=m_a(T)$ has a strict incentive to play the myopic action $-1$. This is because independent of her action in $T$, at the end of $T$ the type $s_a'$ observes whether $s_b\geq s_a$ or $s_b\leq s_a$, and therefore starting from the period $T+1$ she has a dominant action (depending on $z_b(T)$). Symmetrically, the type $s_b'=-m_a(T)$ has a strict incentive to play the myopic action $+1$ for the similar reason. The types arbitrarily close to $s_a'$ and $s_b'$ that are expected to play non-myopically, will similarly deviate in the period $T$ to the respective myopic actions. Since $T$ could be any period following $t=0$, the play is non-myopic symmetric threshold strategies is not incentive-compatible for the players.

(2) Consider symmetric two-threshold strategies.

First, consider $t=0$ and suppose that player $a$'s strategy uses two thresholds $\{\mu_a^1(0),\mu_a^2(0)\}$ such that she chooses the same action $z_a'$ for $s_a<\mu_a^1(0)$ and $s_a>\mu_a^2(0)$, and the other action $-z_a'$ for $s_a\in (\mu_a^1(0),\mu_a^2(0))$; the specification at $s_a=\mu_a^1(0)$ and $s_a=\mu_a^2(0)$ is arbitrary (player $b$'s strategy is symmetric). Assume $z_a'=-1$ (the argument for $z_a'=+1$ is fully symmetric). Observe that $\text{Pr}(s_b\leq -s_a|s_a)\xrightarrow[s_a\rightarrow\infty]{ }0$. I.e., as $s_a$ increases, the probability that player $a$ assigns to the event that the signal of the other player is such that the myopic action is $-1$, goes to 0. At the same time, $E(s_a,S_b(0))>0$ increases in $s_a$. Thus, for any $\delta_a<1$ there exists $s_a'>0$ large enough such that all $s_a>s_a'$ have a strict incentive to choose $+1$ in $t=0$. Thus, all $s_a>\max\{\mu_a^2(0),s_a'\}$ deviate from the prescribed strategy and choose $+1$ in $t=0$.

Second, let $T>0$ be the first period with non-myopic play. Suppose that no agreement has been reached at the end $T-1$. Assume $s_a<0<s_b$ (the other case is completely symmetric). By part (1) of the proof we know that the play in non-myopic symmetric thresholds is not incentive-compatible. Therefore, consider player $a$ and two (history-dependent) thresholds $\{\mu_a^1(T),\mu_a^2(T)\}$ such that player $a$ chooses the same action on $[\inf S_a(T),\mu_a^1(T))$ and $(\mu_a^2(T),\sup S_a(T)]$ (a symmetric strategy applies to player $b$). Observe, however, that at the beginning of the period $T$ the belief sets are mirror images of each other: $S_a(T)=- S_b(T)$. This implies that the type $s_a=\inf S_a(T)$ has a strictly dominant action $-1$ and the type $s_a=\sup S_a(T)$ has a strictly dominant action $+1$. Same applies to the boundary types of the player $b$. Therefore, play in symmetric two-threshold strategies is not incentive-compatible for the players.

\end{proof}

\begin{definition}
A strategy profile leads to a complete aggregation of private information if asymptotically (as $t\rightarrow\infty$) with the exception of measure-zero cases\footnote{Under the standard Lebesgue measure.} all player types take the same action as what they would have chosen under complete information.
\end{definition}

\begin{proposition}
Consider symmetry and a game with discounting, $\delta_a\in(0,1),\delta_b\in(0,1)$. In the class of threshold strategies only the myopic equilibrium leads to a complete aggregation of private information.
\label{aggreg}
\end{proposition}

The proof, which is is relegated to the Appendix, is based on the following observations. First, under the myopic play disagreement cannot continue forever: this holds by Observation 3 $(iii)$. Second, agreement is based on weakly dominant strategies: all player types choose the same actions under myopic agreement as the actions they would have chosen under complete information. In contrast, the play in non-myopic thresholds fails to fully aggregate private signals: there always exist measurable sets of player types who choose a wrong action upon agreement. 

\section{Asymmetry}
Our setting is \emph{asymmetric} if it is not symmetric.  Of course, that admits a huge variety of cases. We will impose the following generic condition on an asymmetric model. Note that in a symmetric setting, $-\s_a (\R_+) = \s_b (\R_-)$, and in particular, $E(\s_a (\R_+), \s_b (\R_-)) = 0$. By an asymmetric model, we refer to \emph{any} situation in which
\begin{equation}
E(\s_a (\R_+), \s_b (\R_-)) \neq 0.
\label{eq:neq}
\end{equation}
\begin{theorem}
In an asymmetric model, myopic play is never an equilibrium in the infinite setting.
\label{th:antim}
 \end{theorem}
\begin{proof}
We will work throughout with the case $s_a < 0$ and $s_b > 0$, presuming that myopic play is adhered to by one of the players, and showing that a profitable deviation is available for some interval of the other player's types, at date 1. 

Suppose that we have myopic play at date 0.
Then at date 1, the relevant configuration is given by
$\{ S_a(1), S_b(1), m_a(1), m_b(1)\} = \{ \R_+, \R_-, m_a, m_b\}$, 
where 
$E(m_a, \R_+) = 0$ and  $E(\R_-, m_b)=0$, so that $m_a = \s_a (\R_+) < 0$ and $m_b = \s_b (\R_-)>0$. The values $m_a$ and $m_b$ are depicted in Figure \ref{fig:asymm}, along with other constructions soon to follow.

\begin{figure}[t!]
\hspace*{-0.1in}
\subfloat{
\includegraphics[width=0.47\textwidth]{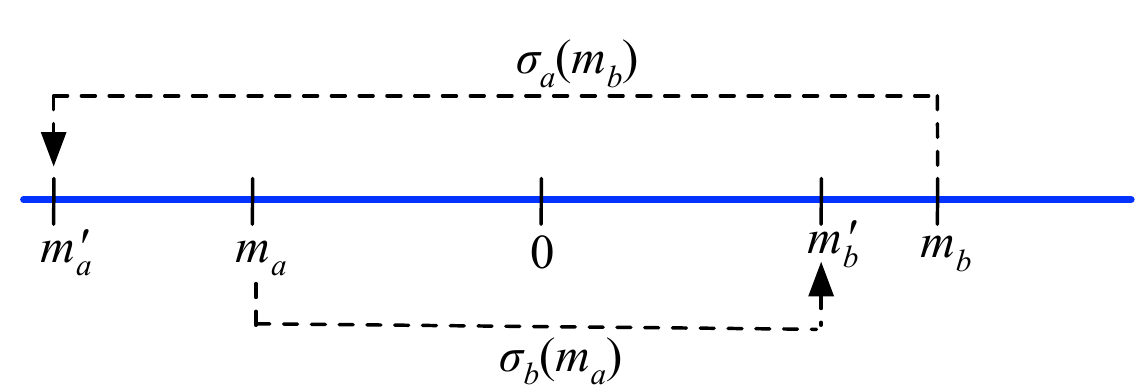}} \hspace*{0.1in}
\subfloat{
\includegraphics[width=0.5\textwidth]{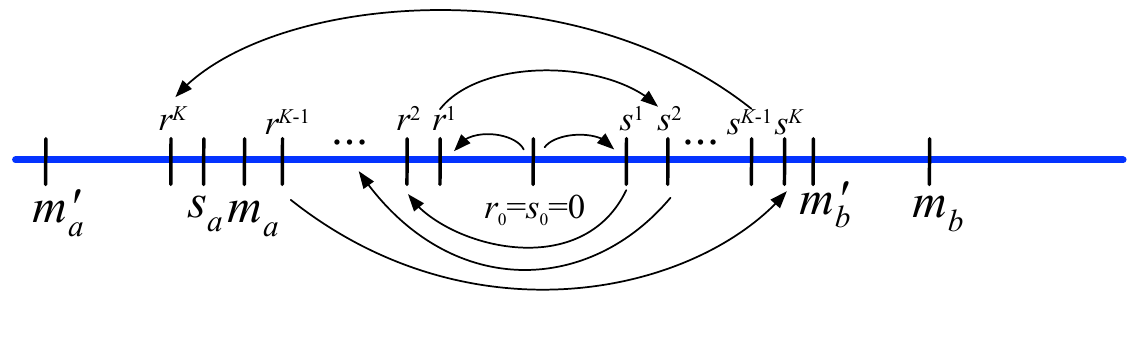}}
\caption{An Illustration of the proof of Theorem \ref{th:antim}.}
\label{fig:asymm}
\end{figure}

Under asymmetry, (\ref{eq:neq}) holds and $E(m_a, m_b) \neq 0$.
Define $(m'_a, m'_b)$ by $E(m'_a, m_b) = E(m_a, m'_b) = 0$; then $m'_a \neq m_a$. Let us suppose that $m'_a < m_a$ (the opposite inequality has a parallel argument indicated at the end of the proof). Then, because $E$ is increasing in each argument, it must be that $m'_b < m_b$. See {left panel of} Figure \ref{fig:asymm}. 

For $r \in [m_a, 0]$, define $g(r)$ by $\s_b ([m_a, r])$, and for $s \in [0, m'_b]$, define $h(s)$ by $\s_a ([s, m_b])$. Because $m_b' = \s (m_a)$, it must be that $g(r) < m'_b$ as long as $r > m_a$. On the other hand, because $E(m_a, [m'_b, m_b]) > E(m_a, m'_b)=0$, we have $h(s) < m_a$ for $s$ close enough to $m'_b$. Therefore if we set $(r^0, s^0) = (0,0)$ and then recursively define  $(r^{k+1}, s^{k+1}) = (h(s^k), g(r^k))$, the first index $K$ at which this sequence wanders out of $[m_a, 0] \times [0, m'_b]$ must be one at which $r^K < m_a$ (while $s^K$ is still below $m'_b$).  It should be noted that $K \ge 2$, because $m_a = \s_a (\R) < \s_a([0, m_b]) = h(0)$. Fix any $T\ge K+1$.

For $\e > 0$, define $s_a = m_a - \e $. Obviously, for $\e > 0$ but small enough,   
\begin{equation}
s_a > m'_a \mbox{ and } s_a > r^K.
\label{eq:sam}
\end{equation}
The right panel of Figure \ref{fig:asymm}. depicts this construction. Now 
for any $\e >0$, let $\pi (\e)$ be the date 1 expected payoff loss to type $s_a=m_a-\e$ from playing $+1$ as opposed to $-1$. Of course,  $\pi (\e) \to 0$ as $\e \to 0$. Now consider two possibilities regarding the realization of $b$'s type:


(I) $s_b > m_b$, say with probability $q \in (0,1)$, computable from model parameters. Then $b$ plays +1 at date $1$, so $S_b(t) \subseteq [m_b, \infty )$ for all $t \ge 2$. Because $s_a > m'_a = \s_a (m_b)$ (see (\ref{eq:sam})), $a$ and $b$ will play $+1$ from date $2$ onward, \emph{no matter what $a$ has played at date $1$}. Therefore, the payoff loss to $a$ from choosing $+1$ as opposed to the myopic best response $-1$, at date 1,  is precisely $\pi (\e)$ in this sub-case.

(II)  $s_b < m_b$, with probability $[1-q] \in (0,1)$ In this case, myopic play by either type dictates -1 (remember $s_a < m_a$). If played, then by  Observation \ref{obs:myopic}(i), both parties will play $-1$ thereafter.

Suppose that type $s_a$ deviates and plays $+1$ instead. Now $b$ will believe (erroneously) that  $S_a(2) = [m_a, 0]$, while because $b$ has played her myopic optimum $-1$, $a$ will believe (correctly) that $S_b(2) = [0, m_b]$.  Consider a continuation strategy in which $a$ plays myopic optima from date $2$ onward, while $b$ as already assumed plays myopic best responses throughout. We wish to compare the payoff from this deviation with the play of $-1$ throughout, as already established under the myopic equilibrium. To do so, note the following properties of myopic play from this stage on:

(i) For $1 \le k < K$, all types $s'_a \le h(s^k)$ \emph{can be rationally expected by} $b$ to be playing $-1$ at date $1+k$, while the remaining possible types $s'_a \in (h(s^k), 0]$ play +1. This perception does not contradict on-path play in the eyes of $b$, because $0 > h(s^k) \ge m_a$ for $k < K$, by construction of $K$.\footnote{A non-generic possibility is $h(s^{K-1})$ \emph{equal} to $m_a$, but even here there is no contradiction to perceived on-path play, as type $m_a$ could be rationally expected to play $-1$.}

(In particular, the type $s_a$ plays $-1$ at date 2, because $h(0) > m_a>s_a$ as already discussed.)

(ii) For $1 \le k < K$, all types $s'_b > g(r^k)$ will play (and will be rationally expected by $a$ to play) $+1$ at date $1+k$, while the remaining possible types $s'_a \in [0, g(r^k)]$ play -1. This play is trivially on-path as $b$ never deviates from myopic play. 

(iii)  Every $s_b$ in $[m'_b, m_b]$ will play $+1$ throughout this phase, because $m'_b > g(r^k)$ for all $1\le k < K$, as already discussed.

In terms of actual play over rounds $1, \ldots , K-1$ (or dates $2, \ldots , K$) ,  type $s_a$ will play $-1$ throughout. Some $b$-types play $-1$ throughout, some play $+1$ throughout, and the remainder choose $+1$ for some initial rounds, then switch to playing $-1$. \emph{As soon as} $b$ plays $-1$, there is agreement, and continuation play must be $-1$ till the end of the game, by  Observation \ref{obs:myopic}(i). 

Otherwise, if $b$ plays $+1$ through these rounds, then it is established that $s_b \ge s^K = g(r^{K-1})$ for in $b$'s perception, all types of $a$ that are larger than $r^{K-1}$ have been eliminated by the observation that $a$ has played $-1$ throughout. But then, because $s_a > r^K$ by (\ref{eq:sam}), it must be that $s_a$ plays $+1$ at round $K$ (or date $K+1$). Moreover, every type $s_b \in [m'_b, m_b]$ will choose $+1$ as well. From this point on to the end of the game, continuation play must be $+1$ by  Observation \ref{obs:myopic}(i). 

Notice that under this last event, there is a payoff gain to player $a$ from choosing $+1$ as opposed to $-1$, under the additional information that she has received by virtue of her deviation at date $1$. Denote this gain by $\Delta (s_a)$. It will occur with probability at least $q'$, where $q' $ is the probability that $s_b \in [m'_b, m_b]$, conditional on $s_b \in [0, m_b]$.   

Recalling that $s_a = m_a - \e$, the expected gain from player $a$'s deviation, discounted to date $1$, and aggregating over both cases I and II, is therefore bounded below by
\[
-\pi (\e ) + \delta^K q'(1-q) \Delta (s_a) = -\pi (\e ) + \delta^K q'(1-q) \Delta (m_a-\e).
\]
It is obvious that $\Delta (m_a-\e)$ is bounded away from 0 in $\e $,\footnote{The gain $\Delta $ to type $m_a$ from knowing that $s_b \in [m'_b, m_b]$ is strictly positive, because myopic play cannot help her achieve this additional, useful information. As $e \to 0$, $\Delta (m_a-\e) \to \Delta$.}  while as already observed, $\pi (\e) \to 0$ as $\e \to 0$.  Therefore there is an interval of types for which player $a$ enjoys a profitable deviation from myopic play.

We end the proof by noting that if $m'_a > m_a$, then $m'_b > m_b$, and the entire argument now works in mirror image by choosing a suitable interval of $b$'s types between $m_b$ and $m_b + \e$.

\end{proof}

The proof reveals that players' types close to the myopic threshold have a strict incentive to deviate: their losses from a deviation can be held arbitrarily close to 0, while with a strictly positive probability they expect to obtain additional information about the types of the other player -- the information that they would forgo otherwise, would they adhere to the myopic play.

Importantly, by a deviation from the myopic play a player has to wait a number of periods before obtaining useful information about the types of the other player. This number of periods becomes larger the more similar the players are. In this sense, for sufficiently similar players myopic equilibrium exists in finite asymmetric games, but it never exists in asymmetric games with the infinite duration.

\section{Conclusions}

This paper characterizes the equilibrium behavior of forward-looking players in a simple two-player repeated-action setting of social learning. In the symmetric environment myopic equilibrium always exists for any positive discount factors. We showed that this equilibrium is unique in symmetric threshold strategies and under simplest non-monotonic
strategies. In the symmetric setting the players cannot disagree forever. Once the agreement is reached, the players choose correct actions -- the same actions that they would choose under complete information. Therefore, myopic play fully aggregates dispersed private information. We have shown that any other threshold strategy fails to fully aggregate private information.

The matters are different when the setting is asymmetric. Here, there is a learning potential even if an agreement is reached. If a player deviates from the myopic play and pretends to be a different type, while the other player presumes that the play continues to be on path, the deviating player expects to gain additional information from the other player with a strictly positive probability.

There are multiple venues for future exploration. First, are there any equilibria in the symmetric setting in arbitrary symmetric non-threshold strategies, and do equilibria in non-symmetric strategies exist? Allowing for arbitrary non-threshold strategies would allow for rich possibilities of codifying private signals, but the question of incentive-compatibility remains open.

Second, what is the solution of a planner in a symmetric setting, when the planner aims to maximize the overall welfare? It would be interesting to know whether the resulting allocation improves upon the decentralized myopic equilibrium.




\begin{spacing}{1.3}

\bibliographystyle{aea}

\bibliography{bib_learning}

\end{spacing}


\section*{Appendix}

\noindent\textit{Illustration of the model using normal distributions}

Suppose that $x$ is drawn from $N(x_0,\sigma^2)$ with $x_0=0$. The agents observe signals $s_a=x+\epsilon_a$, $s_b=x+\epsilon_b$, where $\epsilon_a\sim N(0,\sigma_a^2)$ and $\epsilon_b\sim N(0,\sigma_b^2)$. Therefore, the corresponding distributions are
\[x\sim N(x_0,\sigma^2),\ s_a|x\sim N(x,\sigma_a^2),\ s_b|x\sim N(x,\sigma_b^2).\]


We first derive $e(s_1,s_2)$, the posterior expectation of $x$, using Bayes rule: 
\[
f(x|s_a,s_b)=\frac{f(x)g_a(s_a|x)g_b(s_b|x)}{\int_{x=-\infty}^{+\infty} f(x)g_a(s_a|x)g_b(s_b|x) dx}=
\]
\[
\frac{\frac{e^{\frac{-(x-0)^2}{2\sigma^2}}}{\sqrt{2\pi \sigma^2}}\frac{e^{\frac{-(s_a-x)^2}{2\sigma_a^2}}}{\sqrt{2\pi \sigma_a^2}}\frac{e^{\frac{-(s_b-x)^2}{2\sigma_b^2}}}{\sqrt{2\pi \sigma_b^2}}}{\int_{x=-\infty}^{+\infty}\frac{e^{\frac{-(x-0)^2}{2\sigma^2}}}{\sqrt{2\pi \sigma^2}} \frac{e^{\frac{-(s_a-x)^2}{2\sigma_a^2}}}{\sqrt{2\pi \sigma_a^2}}\frac{e^{\frac{-(s_b-x)^2}{2\sigma_b^2}}}{\sqrt{2\pi \sigma_b^2}} dx}=
\frac{e^{\frac{-(x-\frac{\sigma^2(s_a\sigma_b^2+s_b\sigma_a^2)}{\sigma_a^2\sigma_b^2+\sigma^2(\sigma_a^2+\sigma_b^2)})^2}{2\left( \frac{1}{\sigma^2}+\frac{1}{\sigma_a^2}+\frac{1}{\sigma_b^2}\right)}}}{\sqrt{2\pi \left(\frac{1}{\sigma^2}+\frac{1}{\sigma_a^2}+\frac{1}{\sigma_b^2}\right)}}.
\]
\medskip
The mean of the above distribution is $
e(s_a,s_b)=\frac{\sigma^2(s_a\sigma_b^2+s_b\sigma_a^2)}{\sigma_a^2\sigma_b^2+\sigma^2(\sigma_a^2+\sigma_b^2)},
$
and the variance is $
\frac{1}{\sigma^2}+\frac{1}{\sigma_a^2}+\frac{1}{\sigma_b^2}.
$
An alternative way of formulating the posterior mean of $x$ is to assume that $a$ has a posterior distribution over $s_b$. At the beginning of the game, $a$ observes $s_a$. She then forms a posterior belief about the distribution of $x$, that is
\[
f(x|s_a)=\frac{f(x)g_a(s_a|x)}{\int_{x=-\infty}^{+\infty}f(x)g_a(s_a|x)dx}=e^{\frac{-(x-\frac{s_a\sigma^2}{\sigma^2+\sigma_a^2})^2}{2\frac{\sigma^2\sigma_a^2}{\sigma^2+\sigma_a^2}}}{\sqrt{2\pi \frac{\sigma^2\sigma_a^2}{\sigma^2+\sigma_a^2}}}.
\]
Now, we can derive the posterior distribution $f(s_b|s_a)$ at the beginning of $t=0$ i.e., after $s_a$ is observes by player $a$, and no action of player $b$ has been observed.
\[
\phi(s_b|s_a)=\int_{x=-\infty}^{+\infty}f(s_b|x)f(x|s_a)dx=\frac{e^{\frac{-(s_b-\frac{s_a\sigma^2}{\sigma^2+\sigma_a^2})^2}{2\frac{\sigma_a^2\sigma_b^2+\sigma^2(\sigma_a^2+\sigma_b^2)}{\sigma^2+\sigma_a^2}}}}{\sqrt{2\pi \frac{\sigma_a^2\sigma_b^2+\sigma^2(\sigma_a^2+\sigma_b^2)}{\sigma^2+\sigma_a^2}}}.
\]
So, for example, in $t=0$ the expected value of the state is:
\[
E[x|s_a,s_b\in\mathbbm{R}]=\int_{s_b=-\infty}^{+\infty}\frac{\sigma^2(s_a\sigma_b^2+s_b\sigma_a^2)}{\sigma_a^2\sigma_b^2+\sigma^2(\sigma_a^2+\sigma_b^2)}\frac{e^{\frac{-(s_b-\frac{s_a\sigma^2}{\sigma^2+\sigma_a^2})^2}{2\frac{\sigma_a^2\sigma_b^2+\sigma^2(\sigma_a^2+\sigma_b^2)}{\sigma^2+\sigma_a^2}}}}{\sqrt{2\pi \frac{\sigma_a^2\sigma_b^2+\sigma^2(\sigma_a^2+\sigma_b^2)}{\sigma^2+\sigma_a^2}}}ds_b=s_a\frac{\sigma^2}{\sigma^2+\sigma_a^2}.
\]
Finally, suppose that player $a$, who observes $s_a$, believes that $s_b\in S\subset\mathbbm{R}$. Then, the posterior distribution over $s_b\in S$ is $\frac{f(s_b|s_a)}{\int_{s_b\in S}f(s_b|s_a)ds_b}$, and therefore the expected value of the state $x$ is
\[
E[s_a,\phi(s_b\in S|s_a)]=\int_{s_b=-\infty}^{+\infty}\frac{\sigma^2(s_a\sigma_b^2+s_b\sigma_a^2)}{\sigma_a^2\sigma_b^2+\sigma^2(\sigma_a^2+\sigma_b^2)}\frac{f(s_b|s_a)}{\int_{s_b\in S}f(s_b|s_a)ds_b}ds_b.
\]

We make the following observations. First, note that for any non-empty $S\subset \mathbbm{R}$, $E_{s_a}[s_b|s_a,\phi(s_b\in S|s_a)]=k_{s_a}(S)$, where $k_{s_a}(S)>-\infty$ if $s_a>-\infty$, and $k_{s_a}(S)<+\infty$ if $s_a<+\infty$. But then, if $s_a\rightarrow +\infty$, then for any $S$, $\lim_{s_a\rightarrow+\infty }k_{s_a}(S)>0$ and if $s_a\rightarrow -\infty$, then for any $S$, $\lim_{s_a\rightarrow+\infty }k_{s_a}(S)<0$. This shows that $E[1]$ holds in the setup with normal distributions.

\medskip

 \noindent\textbf{Proof of Proposition \ref{aggreg}}.
Under myopic play, recall that (1) by Observation 3 $(iii)$ disagreement cannot continue forever (apart from measure-zero case of $s_a=-s_b$), and (2) under agreement the play is in weakly dominant strategies. As a result, eventually the players choose correct actions i.e., they choose the same actions as the ones they would choose under complete information.

By Proposition \ref{sym_thresh} we know that the play in symmetric threshold strategies cannot be an equilibrium. Thus, in the following we focus on non-symmetric threshold strategies.

We proceed with the following Claim.

\begin{claim}
Suppose that player $i$ uses a (history-dependent) threshold strategy $m_i(t_0)<\sigma_i(S_j(t_0))$ in some period $t_0\geq 0$, and consider the action choice $z_i(t_0)=+1$. Then, there exists a future period $T_i$ satisfying $t_0<T_i<\infty$, and a path of play from $t_0$ to $T_i-1$ such that in the period $T_i$ player $j$ uses a (history-dependent) threshold strategy $m_j(T_i)$ that satisfies $-m_j(T_i)>m_i(T_i)$, such that following the action $z_j(T_i)=+1$ the agreement is reached, in which case players choose a constant action forever after.

The claim for $m_i(t_0)>\sigma_i(S_j(t_0))$ is completely symmetric.
\label{cl1}
\end{claim}

\noindent \textbf{Proof of Claim \ref{cl1}}. For concreteness, consider the threshold strategy $m_a(t_0)<\sigma_a(S_b(t_0))$.

\noindent \textit{Observation 1}: After player $b$ observes $z_a(t_0)=+1$, for all $s_b\geq -m_a(t_0)$ the action $+1$ is dominant.

\noindent \textit{Observation 2}: $\sup S_b(t_0+1)>-m_a(t_0)$.

To see why Observation 2 holds, note first that there must exist $s_b=-m_a(t_0)\in S_b(t_0)$, as otherwise either $\sup S_b(t_0)<-m_a(t_0)$ or $\inf S_b(t_0)>-m_a(t_0)$ in which cases $s_a=m_a(t_0)$ has a dominant action and is never indifferent between the two actions. Further, since $m_a(t_0)=+1$ is non-myopic, we cannot have $|m_b(t_0)|=m_a(t_0)$ since otherwise $s_a=m_a(t_0)$ learns all payoff-relevant information at the end of $t_0$ independent of her own action in $t_0$, which means that she never chooses non-myopic action $+1$. But then, \textit{if} there exists $m_b(t_0) \in \text{Int}\ S_b(t_0)$, it must be that either $m_b(t_0)<-m_a(t_0)$ in which case following $z_b(t_0)=+1$, $\sup S_b(t_0+1)>-m_a(t_0)$, or $m_b(t_0)>-m_a(t_0)$ in which case following $z_b(t_0)=-1$, $\sup S_b(t_0+1)>-m_a(t_0)$. If there is no $m_b(t_0)\in \text{Int}\ S_b(t_0)$, then since $-m_a(t_0)\in S_b(t_0)$, it must be true that $\sup S_b(t_0)=S_b(t_0+1)>-m_a(t_0)$. This proves Observation 2. 
 
\noindent \textit{Observation 3}: Since $z_a(t_0)=+1$ is non-myopic for all types $s_a\in[m_a(t_0),\sigma_a(S_b(t_0)))$, there must exist a future time period $t_1>t_0$, and a play from $t_0$ to $t_1-1$, such that in the period $t_1$ player $b$ uses a (history dependent) threshold $m_b(t_1)\in\text{Int}\ S_b(t_1)$. If no such play exists (i.e. if $s_a$ does not expect any payoff-relevant information from $b$ in any future period), then $a$ will never make a non-myopic choice in $t_0$.

From Observations 1-3 it follows that (1) there must exist a period $t_1>t_0$ with a strictly interior threshold $m_b(t_1)\in\text{Int}\ S_b(t_1)$ that satisfies $-m_b(t_1)>m_a(t_0)$ and that (2) following $z_a(t_0)=+1$, player $b$ never uses thresholds within the set $[-m_a(t_0),\sup S_b (t_0+1)]$ (where $-m_a(t_0)<\sup S_b(t_0+1)$).

For the following assume $m_a(t_0)>0$ (the proof for $m_a(t_0)<0$ is entirely parallel). Define a sequence of (history dependent) thresholds for player $b$, $ q_b:=(m_b(t_1),m_b(t_2),..)$, such that $t<t_1<t_2<..$ and $m_b(t_1)<m_b(t_2)<..<m_b(t_j)<..$, with the understanding that if both $m_b(t_k)\in q_b$ and $m_b(t_{k+1})\in q_b$, then starting from the period $t_k$ there exists a play up to the period $t_{k+1}-1$ resulting in the strategy $m_b(t_{k+1})$ in the period $t_{k+1}$, with $m_b(t_k)<m_b(t_{k+1})$. By Observation 2, note that $q_b$ has at least one element (and so $q_b$ is non-empty), and by Observation 1 note that none of the elements of $q_b$ exceeds $-m_a(t_0)$.



\noindent \textit{Observation 4}. If for the type $s_a=m_a(t_0)$ the agreement is reached at the end of a finite period $T_a$, then there must exist the largest element in the sequence $q_b$, some $m_b(T_a)$, such that upon $z_b(T_a)=+1$ the players agree on the same action forever after.

\noindent \textit{Observation 5}. Independent of whether $q_b$ is finite or infinite, there must exist a play from $t_0$ to $T'>t_0$, with $m_b(T')\in q_b$, such that where following $z_b(T')=+1$, $\sigma_a(S_b(T'+1))<m_a(t_0)$.

Observation $4$ holds directly by construction. To see why Observation $5$ must hold, suppose to the contrary that the statement is not true. This would imply that there is no element $m_b(t')\in q_b$ such that following $z_b(t')=+1$, $\sigma_a(S_b(t'+1))<m_a(t_0)$. Since after observing $z_b(t')=-1$ the type $s_a'=m_a(t_0)$ has a dominant action $-1$, this would imply that after the period $t_0$ with $z_a(t_0)=+1$, no action of $b$ ever leads to the action $+1$ being myopically strictly optimal for the type $s_a'$. Now, consider an infinite sequence $\{z_a(t')=-1\}_{t'\geq t_0}$. Since there is no play following $z_a(t_0)=+1$ where $+1$ is myopically optimal for the type $s_a'$, the sequence $\{z_a(t')=-1\}_{t'\geq t_0}$ must dominate any other (on path) sequence of actions for $s_a'$. This means that $z_a(t_0)=+1$ cannot be optimal for $s_a'$, a contradiction.

Now, equipped with the above Observations, suppose that contrary to the Claim a period $T_a<\infty$ does not exist. That is, the sequence $q_b$ must have infinitely many elements. Because the sequence $q_b$ is monotonically increasing, and by Observation 1 it is bounded above by $-m_a(t_0)$, by Monotone Convergence Theorem there must exist a number $r\leq -m_a(t_0)$ to which the sequence $q_b$ converges.

Next, consider any $m_b(t')\in q_b$. If $q_b$ is infinite, there must exist a play and a future time period $t''>t'$ and a threshold $m_b(t'')\in q_b$ with the property $m_b(t')<m_b(t'')$. 

\noindent\textit{Observation 6}: Following $z_b(t')=+1$ there must exist a play and a period $\hat{t}>t'$ where $a$ uses a threshold $m_a(\hat{t})$ with the property $m_a(\hat{t})<-m_b(t')$.

Suppose not. Since any future threshold of $a$ must be strictly below $-m_b(t')$ because after the period $t'$ all $s_a\geq -m_b(t')$ have the dominant action $+1$, this would imply that following the move $z_b(t')=+1$, none of the types $s_b\in[m_b(t'),\sup S_b(t')]$ will ever receive an additional information from the player $a$: i.e., $S_a(t')=S_a(t'+1)=..$. But then, either $b$'s action $+1$ at the threshold $m_b(t')$ or one of the actions at the threshold $m_b(t'')$ is suboptimal.

Recall that by Observation 4, there exists a period $T'> t_0$ and a threshold $m_b(T')\in q_b$, such that following $z_b(T')=+1$, $\sigma_a(S_b(T'+1))<m_a(t_0)$. 
Thus, by Observation 6 the necessary condition for $q_b$ to be infinite, is that there is a course of play (on path) where the type $s_a'=m_a(t_0)$ chooses the non-myopic action $-1$ infinitely many times. But as $\lim_{t\rightarrow\infty}q_b=r$, the value from receiving additional information from $b$ for the type $s_a'$ goes to 0: this is because $b$'s type set for which $a$ switch the myopically optimal action shrinks, and goes to 0 in length. At the same time, after the period $T'$ (Observation 5) the myopic losses of $s_a'$ from choosing the non-myopic action $-1$ increase in absolute value. This is the consequence of Observation 5: 
for any two elements $m_b(t')\in q_b,m_b(t'') \in q_b$, with $T'<t'<t''$, it must be true that following $z_b(t')=+1$, $\sigma_a(S_b(t'+1))<m_a(t_0)$, and following $z_b(t'')=+1$, $\sigma_a(S_b(t''+1))<\sigma_a(S_b(t'+1))$. 

But then, if there is no largest element in $q_b$, for any $\delta_a<1$ there exists a course of play and a time period $\tilde{t}$ such that for all $t>\tilde{t}$ the type $s_a'$ chooses the non-myopic action $-1$ although the discounted expected value of additional information for the type $s_a'$ is strictly smaller than the current-period myopic loss. A contradiction. As a result, $q_b$ must be finite and therefore must have the largest element. Since all types $s_b\geq -m_a(t_0)$ have the strictly dominant action $+1$ upon observing $z_a(t_0)=+1$, it must be that $\max(q_b)<-m_a(t_0)$.
This proves the claim.

\medskip

Finally, we consider the following two observations. First, under agreement in a period $T$ the actions of both players starting at $T$ must be myopically optimal. Thus, using Claim \ref{cl1}, the action $m_b(T)=+1$ must be myopically optimal. Note that for $m_b(T)$ to result in agreement, it must be the case that even if $a$ uses an interior threshold $m_a(T)\in\text{Int}\ S_a(T)$ in the same period $T$, it must be that $m_a(T)>-m_b(T)$ as otherwise, upon $z_a(T)=-1$, $\sigma_b(S_a(T+1))\in \text{Int}\ S_b(T+1)$ which contradicts the notion of an agreement. Therefore, it must be the case that following $z_b(T)=+1$ both players choose the myopically optimal action $+1$ for all $t>T$. That is, both $\sigma_a(S_b(T+1))<m_a(t_0)$ and $\sigma_b(S_a(T+1))\leq m_b(T)$ must hold (the first inequality is strict by the argument in Claim \ref{cl1}). 

 Note, however, that for all $(s_a,s_b)\in (m_a(t_0),-m_b(T))\times (m_b(T),-m_a(t_0))$ we have $E(s_a,s_b)<0$. That is, since $m_a(t_0)< -m_b(T)$ there exists sets (of positive measure) of types for each player on which the players are taking a different action to what they would have chosen under complete information.

Second, note that as we assume play in non-myopic threshold strategies, there must exist at least one history, following which $i$ uses a strategy $m_i(t)<M_i(t)$. This proves the Proposition.


\qed

\end{document}